\documentclass{amsart}

\usepackage{amsthm}
\usepackage{amsmath}
\usepackage{amsfonts}
\usepackage{mathtools}
\usepackage{accents}
\usepackage{tikz}
\usepackage{comment}
\usepackage{enumitem}
\usepackage{cleveref}

\newtheorem{lemma}{Lemma}

\newtheorem{theorem}[lemma]{Theorem}

\newtheorem{definition}[lemma]{Definition}

\newcommand{\reals}{\mathbb{R}}
\newcommand{\naturals}{\mathbb{N}}

\newcommand{\lowtheta}{\theta_*}
\newcommand{\lowt}{\tau}
\newcommand{\var}{\mathrm{Var}}
\newcommand{\cov}{\mathrm{Cov}}
\newcommand{\est}{{\hat{\theta}}}
\newcommand{\estlow}{{\hat{\theta}_*}}
\newcommand{\estupp}{{\hat{\theta}^*}}
\newcommand{\gambles}{\mathcal{L}}

\newcommand{\estsigma}{\hat{\sigma}}
\newcommand{\estvar}{\hat{\sigma}^2}

\usepackage{url}

\title[Imprecise Monte Carlo simulation and iterative importance sampling]{Imprecise Monte Carlo simulation and iterative importance sampling for the estimation of lower previsions}
\author{Matthias C. M. Troffaes}
\address{Durham University, Department of Mathematical Sciences, UK}
\email{matthias.troffaes@durham.ac.uk}


\begin{document}

\begin{abstract}
  We develop a theoretical framework for studying numerical estimation of lower previsions, generally applicable to two-level Monte Carlo methods, importance sampling methods, and a wide range of other sampling methods one might devise. We link consistency of these estimators to Glivenko-Cantelli classes, and for the sub-Gaussian case we show how the correlation structure of this process can be used to bound the bias and prove consistency. We also propose a new upper estimator, which can be used along with the standard lower estimator, in order to provide a simple confidence interval. As a case study of this framework, we then discuss how importance sampling can be exploited to provide accurate numerical estimates of lower previsions. We propose an iterative importance sampling method to drastically improve the performance of imprecise importance sampling. We demonstrate our results on the imprecise Dirichlet model.
\end{abstract}

\maketitle

\section{Introduction}
\label{sec:intro}

Various sensible approaches to sampling for estimation of lower previsions
can be found in the literature.
A case study comparing a wide range of techniques,
specifically aimed at reliability analysis,
can be found in \cite{2009:oberguggenberger:sensitivityreview}.

A first approach is to use \emph{two-level Monte Carlo sampling},
  where first one samples distributions over the (extreme points of the) credal set,
  and then samples from these distributions,
In the context of belief functions, one can also \emph{sample random sets}, and then evaluate the resulting belief function through optimisation over these sets
  \cite{1996:moral::importance}.
A third more sophisticated approach comprises of \emph{importance sampling} from a reference distribution, and
  then solve an optimisation problem over the importance sampling weights
  \cite{2009:oneill:importancesampling,2015:fetz:montecarlo,2016:zhang:importancesampling}.

Two-level Monte Carlo can be rather inefficient, especially if a credal set is high-dimensional or if it has a large number of extreme points. Moreover, two-level Monte Carlo generally only provides a non-conservative
solution.

Random set sampling is more efficient, but requires a large number
of optimisation problems to be solved (one for each sample), and
requires a suitable belief function approximation to be identified if
one wants to apply this to arbitrary lower previsions.

Importance sampling in imprecise probability
has been studied already in the '90s; see
for example
\cite{1996:moral::importance,1996:cano::importance,1997:hernandez::importance}
for some early works.
Importance sampling can be quite effective.
For example,
\cite{2015:deangelis:linesampling} have successfully used sensitivity analysis
over importance sampling weights with respect to the mean parameter of a
normal distribution.
In \cite{2015:fetz:montecarlo}, importance sampling is used over both the
mean and the variance parameters of a normal distribution using
a 2-dimensional grid.
A common issue with importance sampling is that estimates can be very poor
if the reference distribution is far off the optimal distribution.
This has been recently addressed in \cite{2017:troffaes::impmc},
where an iterative self-normalised importance sampling method is proposed
that requires far less
computational power compared to standard importance sampling methods
for sensitivity analysis, in the sense that far smaller samples can be
used, and that far smaller optimisation problems need to be solved.
A very similar approach using standard importance sampling was
proposed in \cite{2017:fetz::montecarlo}.

A second issue, which has received little attention in the literature,
concerns the bias and consistency of these estimates.
Two-level Monte Carlo methods and
importance sampling methods are essentially constructed as lower envelopes
of estimators for precise expectations.
To the best of the author's knowledge, in the context of imprecise probability,
the bias and consistency of such lower envelopes has not yet been studied
elsewhere in the literature.

The first purpose of this paper is to develop
a theoretical framework for studying numerical
estimation of lower previsions.
We do so by looking, in essence, at the
estimation of the minimum of an unknown function, given that
we can `simulate' the function to an arbitrary degree of precision.
This framework applies to
two-level Monte Carlo methods, importance sampling methods, and a wide
range of other sampling methods one might devise.
A first contribution of this paper is a link between consistency
of such envelopes and a Glivenko-Cantelli class condition.
In case our estimators can be described by a sub-Gaussian process
(by the central limit theorem,
this will often hold if the sample size is taken large enough),
we show how the correlation structure of this process
can be used to bound the bias.
We provide sufficient conditions for consistency,
and we also identify situations where consistency fails.
We rely heavily on stochastic process theory,
and in particular, on Talagrand's results on the
supremum of stochastic processes \cite{2014:talagrand}.

Although we obtain theoretical bounds, these bounds are not practical
in the sense that they require us to bound a rather tricky
functional. A second theoretical contribution of the paper is
that we propose a new `upper' estimator,
which can be used along with the standard
lower envelope estimator, in order to provide a single comprehensive confidence
interval. Unfortunately, the consistency
of this upper estimator is still an open problem,
although it is guaranteed to provide an upper bound (hence its name).
However, we can identify one situation under which the upper estimator
is unbiased.

A second aim of the paper is to study importance sampling for the estimation
of lower previsions.
We follow \cite{2009:oneill:importancesampling,2017:troffaes::impmc,2017:fetz::montecarlo},
and look specifically at how we can take envelopes
over the importance sampling weights directly in order to obtain
sampling estimates, without needing to draw large numbers of samples,
and without needing to solve large numbers of optimisation problems.
Unlike \cite{2009:oneill:importancesampling}, however, we do not
just look at Bayesian sensitivity analysis, and admit arbitrary sets
of distributions in our theoretical treatment.
Also unlike for instance
\cite{2009:oneill:importancesampling,2015:deangelis:linesampling,2015:fetz:montecarlo,2016:zhang:importancesampling,2017:fetz::montecarlo},
in this paper, we use self-normalised importance sampling
instead of standard importance sampling,
as it is known that this drastically speeds up calculations
\cite{2017:troffaes::impmc}.
In this paper, we will also show that
it ensures coherence of the resulting estimates (regardless of bias).

We then revisit the iterative importance sampling method
proposed in \cite{2017:troffaes::impmc,2017:fetz::montecarlo}.
The key novelty of this method is the idea of iteratively changing the importance
sampling distribution itself, in order to ensure that the final answer
has an effective sample size that is as close as possible to the
actual sample size.
In this paper, we focus on obtaining proper confidence intervals.
We reconfirm that the iterative method requires far less
computational power compared to standard importance sampling methods
for sensitivity analysis, in the sense that far smaller samples can be
used, and that far smaller optimisation problems need to be solved.
We also identify conditions under which we can obtain confidence intervals
almost instantly.

The paper is structured as follows. In \cref{sec:estimation},
we study the theory behind lower envelopes of estimators.
We study bias, consistency, and
two ways to obtain a confidence interval.
In \cref{sec:iteratedimpsamp}, we study importance sampling for the
estimation of lower previsions.
We briefly review the basic theory, and then
we revisit the iterative importance sampling method
proposed in \cite{2017:troffaes::impmc,2017:fetz::montecarlo},
from the perspective of the results derived in the preceding section.
Some numerical examples demonstrate our approach,
using the imprecise Dirichlet model.
We conclude in \cref{sec:conclusion}.

\section{Imprecise Estimation}
\label{sec:estimation}

The aim of this section is to provide a general theoretical framework
for studying the lower envelope of a parameterized set of estimators.
The main idea is that we can use such a lower envelope to estimate
the minimum of an unknown (or, hard to evaluate) function.
The application that we have in mind is one where we have
a parameterized set of estimators for the expectation of some quantity,
and we wish to estimate the lower expectation.
However, the theory developed in this section
is not tied to any specific parameterized set of estimators.

\subsection{Lower and Upper Estimators for the Minimum of a Function}

Let $X$ be a random variable (or, vector of random variables)
taking values in some set $\mathcal{X}$.
Let $t$ be a parameter taking values in some set $\mathcal{T}$.
Let $\est\colon\mathcal{X}\times\mathcal{T}\to\reals$ be an arbitrary function,
such that for every $t\in\mathcal{T}$, $\est(X,t)$ is an unbiased estimator for
$\theta(t)$.
In other words, $\est(X,\cdot)$ is an unbiased
estimator for the function $\theta(\cdot)$.
We explicitly isolate the random part of our estimator by writing it
as a function of a random variable $X$.
This decomposition is essential later, as it will allow us to construct
an upper estimator.

More specifically, we assume that
for every fixed $t$, $\est(\cdot,t)$ is measurable,
\begin{equation}
  E(\est(X,t))=\theta(t),
\end{equation}
and $\var(\est(X,t))$ is finite (and hopefully quite small).

We assume that the function $\theta(t)$ has a minimum:
\begin{equation}
  \lowtheta\coloneqq\min_{t\in\mathcal{T}}\theta(t).
\end{equation}
Our aim is to construct an estimator for that minimum.

In case of estimation of lower previsions,
consider a gamble $f$,
a (weak-*) compact collection $P_t$ of previsions (expectation operators)
parameterized by $t$,
and a collection of estimators $\est(X,t)$ for $P_t(f)$.
With
\begin{align}
  \theta(t)&\coloneqq P_t(f)
\end{align}
we aim to construct estimators for $\underline{P}(f)\coloneqq\lowtheta$ and study the properties
of such estimators.

Throughout, we assume that $\mathcal{T}$ is a compact subset of $\reals^k$,
and that $\est(x,t)$ is continuous in $t$ for all $x$.
This guarantees that $\est(x,t)$ can be minimised over $t$ for any
value of $x\in\mathcal{X}$.
Because $\reals^k$ is separable, there is a countable subset
$\mathcal{T}'$ of $\mathcal{T}$ such that, for all $x$,
\begin{equation}
  \inf_{t\in\mathcal{T}'}\est(x,t)=\min_{t\in\mathcal{T}}\est(x,t).
\end{equation}
Consequently, $\min_{t\in\mathcal{T}}\est(x,t)$ is measurable,
and additionally there is a measurable function $\lowt\colon\mathcal{X}\to\mathcal{T}$
such that
\begin{equation}
  \lowt(x)\in\arg\min_{t\in\mathcal{T}} \est(x,t).
\end{equation}

The next theorem provides us with a lower and upper estimator for
$\lowtheta$.

\begin{theorem}\label{thm:lowuppest}
  Assume $X$ and $X'$ are i.i.d.\ and let
  \begin{align}
    \estlow(X)&\coloneqq \est(X,\lowt(X))=\min_{t\in\mathcal{T}}\est(X,t) \\
    \estupp(X,X')&\coloneqq \est(X, \lowt(X'))
  \end{align}
  Then
  \begin{align}\label{eq:estlowupp:ineq}
    \estlow(X)
    \le
    \estupp(X,X')
  \end{align}
  and
  \begin{align}
    E(\estlow(X))
    \le
    \lowtheta
    \le
    E(\estupp(X,X')).
  \end{align}
\end{theorem}
\begin{proof}
  For all $x$ and $x'\in\mathcal{X}$,
  \begin{equation}
    \estlow(x)=\min_{t\in\mathcal{T}}\est(x,t)
    \le
    \est(x,\lowt(x'))=\estupp(x,x')
  \end{equation}
  This proves the first inequality.

  For the second inequality, note that, for all $t\in\mathcal{T}$,
  \begin{equation}
    \estlow(X)=\min_{t'\in\mathcal{T}}\est(X,t')\le \est(X,t)
  \end{equation}
  Now take the expectation, and then the minimum over $t$, from both
  sides of this inequality, to arrive at
  \begin{equation}
    E(\estlow(X))\le \min_{t\in\mathcal{T}}E(\est(X,t))=\lowtheta
  \end{equation}

  For the final inequality, first note that for all $x'\in\mathcal{X}$,
  by the independence of $X$ and $X'$,
  \begin{equation}
    E(\est(X,\lowt(x'))\mid X'=x')
    =
    E(\est(X,\lowt(x')))
    \ge
    \min_{t\in\mathcal{T}}E(\est(X,t))
    =\lowtheta.
  \end{equation}
  Consequently, by the law of iterated expectation,
  \begin{align}
    E(\estupp(X,X'))
    &=
    E(\est(X,\lowt(X')))
    =
    E(E(\est(X,\lowt(X'))\mid X'))
    \ge
    \lowtheta.
  \end{align}
\end{proof}

The estimator $\estlow(X)$ is used throughout the literature (see \cite{2009:oberguggenberger:sensitivityreview} and references therein) as an
estimator for lower previsions.  In that context, it is not normally
noted in the literature that $\estlow(X)$ is negatively biased, so in
that sense, the above theorem provides a `new' result.  We shall see
that the bias can be very large in specifically constructed examples.
However, provided that our estimator has the form of a sample mean,
we will show that $\estlow(X)$ is a consistent estimator for
$\lowtheta$ if $\{\est(X,t)\}_{t\in\mathcal{T}}$ is a
Glivenko-Cantelli class. If, additionally, $\est(X,t)$ is
(approximately) a Gaussian process over $t\in\mathcal{T}$---this holds if the estimators
satisfy the central limit theorem which will be the case for most
estimators used in practice---then we show
consistency is satisfied if the
process has a finite Talagrand functional, and in that case we can
also explicitly bound the bias as a function of this functional.
Moreover, if for every realisation of $x$,
$\est(x,t)$ is a coherent prevision when seen as a function of the
gamble $f$, then $\estlow(x)$ is guaranteed to be a coherent lower prevision
in itself.

The estimator $\estupp(X,X')$ is novel, as far as we know.
As we shall see, we cannot yet prove much about it.
Its main use is that it allows us to bound the bias of
$\estlow(X)$ without having to explicitly bound the Talagrand functional,
which is a challenging problem for general estimators.
Currently, we do not know the conditions under which $\estupp(X,X')$
is consistent in general.
Further, it is easy to see that $\estupp(x,x')$ is not coherent
in general when seen as a function of $f$.
In this paper, we will simply use $\estupp(X,X')$
as a diagnostic for $\estlow(X)$
to avoid complicated generic chaining arguments to obtain bounds
for the Talagrand functional.

Note that $\est(X,T)$ is a positively biased estimator for $\lowtheta$,
for \emph{any} random variable $T$ taking values in $\mathcal{T}$,
as long as $T$ is independent of $X$. So we do not necessarily have to
take $T=\tau(X')$ as in the theorem,
although this seems the most obvious choice.
The theoretically optimal choice for $T$
is to take $T=\arg\min_{t\in\mathcal{T}}\theta(t)$,
and in that case $\est(X,T)$ is an unbiased estimator for $\lowtheta$.
Finding this optimal choice amounts to calculating
$\lowtheta=\min_{t\in\mathcal{T}}\theta(t)$,
which is the quantity we are aiming to estimate.
Therefore, $T=\arg\min_{t\in\mathcal{T}}\theta(t)$ is not a useful choice.

\subsection{Unbiased Case}

The following theorem states a simple condition under which both
$\estlow(X)$ and $\estupp(X,X')$ estimators are unbiased.

\begin{theorem}
  If there is a $t^*\in\mathcal{T}$ such that $\est(X,t^*)\le \est(X,t)$
  for all $t\in\mathcal{T}$, then
  \begin{align}
    \estlow(X)=\estupp(X,X')=\est(X,t^*)
  \end{align}
  and consequently,
  \begin{align}
    E(\estlow(X))
    =
    \lowtheta
    =
    E(\estupp(X,X')).
  \end{align}
\end{theorem}
\begin{proof}
  Simply note that we may choose $\lowt(x)=t^*$ for all $x\in\mathcal{X}$.
  Now apply \cref{thm:lowuppest}.
\end{proof}

So, in this case, the lower and upper estimators coincide,
and therefore there is no bias.
The theorem provides a reason for choosing $\est(X,\tau(X'))$ as our
upper estimator.
Indeed, if there is a $t^*$ such that $\est(X,t^*)\le \est(X,t)$,
then $\tau(X')$ will identify it.
Normally however, there is no $t^*$ such that $\est(X,t^*)\le \est(X,t)$ for all $t$.

\subsection{Consistency}
\label{sec:consistency}

A good estimator will allow us to control the error,
through the sample size. For example,
an estimator may take the form of a sample mean:
\begin{equation}
  \est_n(X,t)=\frac{1}{n}\sum_{i=1}^n \est(X_i,t)
\end{equation}
where $X\coloneqq (X_i)_{i\in\naturals}$ and where the $X_i$ are i.i.d.\ random
variables taking values in $\mathcal{X}$.
Such estimators are related to so-called empirical processes.

For any fixed $t\in\mathcal{T}$,
we can make the error arbitrary small
because
\begin{equation}
  \var(\est_n(X,t))\propto 1/n.
\end{equation}
An estimator where the error can be made arbitrarily small
is called consistent. More formally \cite[Chapter~6]{2002:schervish:degroot}:

\begin{definition}\label{def:consistent}
  A (sequence of) estimator(s) $Z_n$ for $z\in\reals$
  is called \emph{consistent} whenever,
  for all $\epsilon>0$, we have that
  \begin{equation}
    \lim_{n\to\infty}P(|Z_n-z|>\epsilon)=0.
  \end{equation}
\end{definition}

A natural question is: if $\est_n(X,t)$ is consistent for $\theta(t)$,
will
\begin{equation}
  \estlow_n(X)\coloneqq\min_{t\in\mathcal{T}}\est_n(X,t)
\end{equation}
be consistent for $\lowtheta$?
Even if $\estlow_n(X)$ is biased,
consistency of $\estlow_n(X)$ would help,
because in that case the bias can be made arbitrarily small by increasing $n$.
First, we show that $\estlow_n(X)$ is consistent in case $\mathcal{T}$
is finite, based on a simple union bound.
Then, we consider the infinite case,
linking consistency to so-called Glivenko-Cantelli classes.
Under
the additional assumption that $\est_n(X,t)$ is approximately
normally distributed,
we will also quantify the bias of $\estlow_n(X)$
based on the Talagrand functional.

\begin{theorem}
  Suppose that, for all $t\in\mathcal{T}$, 
  $\est_n(X,t)$ is a consistent estimator for $\theta(t)$.
  If $\mathcal{T}$ is finite,
  then $\estlow_n(X)$ is a consistent
  estimator for $\lowtheta$.
\end{theorem}
\begin{proof}
  Fix $\epsilon>0$. We know that, for all $t\in\mathcal{T}$,
  \begin{equation}
    \lim_{n\to\infty}P(|\est_n(X,t)-\theta(t)|>\epsilon)=0
  \end{equation}
  We need to show that
  \begin{equation}
    \lim_{n\to\infty}P(|\estlow_n(X)-\lowtheta|>\epsilon)=0
  \end{equation}
  Because
  \begin{align}
    P(|\estlow_n(X)-\lowtheta|>\epsilon)
    &=
    P\left(\estlow_n(X)>\lowtheta+\epsilon\right)
    +
    P\left(\estlow_n(X)<\lowtheta-\epsilon\right),
  \end{align}
  it suffices to show that both terms on the right hand side converge to zero.

  Both terms can be easily bounded using Boole's inequality.
  Indeed,
  \begin{align}
    P\left(\estlow_n(X)>\lowtheta+\epsilon\right)
    &=
    P\left(
      \min_{t\in\mathcal{T}}\est_n(X,t)>\min_{s\in\mathcal{T}}\theta(s)+\epsilon\right)
    \\
    &=
    P\left(
      \bigcup_{s\in\mathcal{T}}\left\{
      \min_{t\in\mathcal{T}}\est_n(X,t)>\theta(s)+\epsilon\right\}\right)
    \\
    &\le
    \sum_{s\in\mathcal{T}}
    P\left(
      \min_{t\in\mathcal{T}}\est_n(X,t)>\theta(s)+\epsilon\right)
    \\
    &\le
    \sum_{s\in\mathcal{T}}
    P\left(
      \est_n(X,s)>\theta(s)+\epsilon\right)
  \end{align}
  The last expression converges to zero because $\est_n(X,s)$
  is a consistent estimator for $\theta(s)$.

  Similarly,
  \begin{align}
    P\left(\estlow_n(X)<\lowtheta-\epsilon\right)
    &=
    P\left(
      \min_{t\in\mathcal{T}}\est_n(X,t)<\min_{s\in\mathcal{T}}\theta(s)-\epsilon\right)
    \\
    &=
    P\left(
      \bigcup_{t\in\mathcal{T}}\left\{
      \est_n(X,t)<\min_{s\in\mathcal{T}}\theta(s)-\epsilon\right\}\right)
    \\
    &\le
    \sum_{t\in\mathcal{T}}
    P\left(
      \est_n(X,t)<\min_{s\in\mathcal{T}}\theta(s)-\epsilon\right)
    \\
    &\le
    \sum_{t\in\mathcal{T}}
    P\left(
      \est_n(X,t)<\theta(t)-\epsilon\right)
  \end{align}
  Again, the last expression converges to zero because $\est_n(X,t)$
  is a consistent estimator for $\theta(t)$.
\end{proof}

The case where $\mathcal{T}$ is not finite is considerably more difficult.
A first important insight is that
the problem can be linked to so-called Glivenko-Cantelli classes
for those estimators that take the form of a sample mean (which covers
a large range of estimators).
\begin{definition}\cite[p.~272]{2014:talagrand}
  Consider a sequence $X_1$, $X_2$, \dots of i.i.d.\ random variables
  taking values in $\mathcal{X}$.
  A countable set $\mathcal{F}$ of measurable functions on $\mathcal{X}$
  is called a
  \emph{Glivenko-Cantelli class} if
  \begin{equation}
    \lim_{n\to\infty}E\left(
      \sup_{f\in\mathcal{F}}\left|
        \frac{1}{n}\sum_{i=1}^n (f(X_i)-E(f(X_i)))
      \right|
    \right)
    =
    0
  \end{equation}
\end{definition}
Remember that $\mathcal{T}'$ is a countable dense subset of
$\mathcal{T}$ whose existence is guaranteed under the assumptions made
at the beginning of the paper.
We can now state our first main result.
\begin{theorem}
  If
  the set of functions $\{\est(\cdot,t)\colon t\in\mathcal{T}'\}$
  is a Glivenko-Cantelli class,
  then
  $\estlow_n(X)\coloneqq\min_{t\in\mathcal{T}}\frac{1}{n}\sum_{i=1}^n \est(X_i,t)$
  is a consistent estimator for $\lowtheta$.
\end{theorem}
\begin{proof}
  Fix any $\epsilon>0$. Then, by Markov's inequality,
  \begin{align}
    P(|\estlow_n(X)-\lowtheta|>\epsilon)
    &\le
    \frac{E(|\estlow_n(X)-\lowtheta|)}{\epsilon}
  \end{align}
  so it suffices to show that $E(|\estlow_n(X)-\lowtheta|)$ converges to zero.
  Indeed,
  \begin{align}
    E(|\estlow_n(X)-\lowtheta|)
    &=
    E\left(\left|\inf_{t\in\mathcal{T}'}\est_n(X,t)-\inf_{t\in\mathcal{T}'}\theta(t)\right|\right)
    \\
    &\le
    E\left(\sup_{t\in\mathcal{T}'}\left|\est_n(X,t)-\theta(t)\right|\right)
  \end{align}
  By the Glivenko-Cantelli class assumption, the right hand side converges to zero.
\end{proof}

\subsection{Discrepancy Bounds}
\label{sec:discrepancy:bounds}

For practical reasons, it is useful to theoretically quantify the bias,
in order to gain some intuition
for how we should design a set of estimators that
can achieve a small bias.
A range
of powerful techniques for evaluating so-called discrepancy
bounds is presented in \cite[Chapter~9]{2014:talagrand},
and these can be readily applied to our problem.
However, by the central limit theorem, $\est_n(X,t)$ will approximate a
Gaussian process. Consequently,
the theory for evaluating the supremum of a Gaussian process
\cite[Chapter~2]{2014:talagrand}
is applicable here too, and as discrepancy bounds for Gaussian processes
are much easier to analyse, we will present and apply the key results just
for the Gaussian process case here, referring the reader to the literature
\cite{2014:talagrand} for further (and a lot more technical) treatment of
this fascinating theoretical problem.

For convenience, we introduce:
\begin{equation}
  Z_n(t)\coloneqq \est_n(X,t)-\theta(t)
\end{equation}
We define a pseudo-metric on $\mathcal{T}$ as follows:
\begin{equation}
  d_n(s,t)\coloneqq \sqrt{E\left((Z_n(s)-Z_n(t))^2\right)}
\end{equation}
For any $A\subseteq\mathcal{T}$, let $\Delta_n(A)\coloneqq\sup_{s,t\in A}d_n(s,t)$
denote the diameter of $A$.

We assume that the process $Z_n(t)$ satisfies the following increment condition
\cite[p.~13]{2014:talagrand}:
\begin{equation}\label{eq:increment:condition}
  \forall u>0,\quad
  P(|Z_n(s)-Z_n(t)|\ge u)\le 2\exp\left(-\frac{u^2}{2d_n(s,t)^2}\right).
\end{equation}
This holds if $\est_n(X,t)$ is a Gaussian process.
If \cref{eq:increment:condition} holds, we will say that
$\est_n(X,t)$ is \emph{sub-Gaussian}.
\begin{definition}\cite[p.~25]{2014:talagrand}
  An \emph{admissible sequence} is an increasing sequence $\mathcal{A}_k$
  of partitions of $\mathcal{T}'$ such that the cardinality of $\mathcal{A}_k$
  is $1$ for $k=0$, and less or equal than $2^{(2^k)}$ for $k\ge 1$.
\end{definition}
For any $t\in\mathcal{T}'$, $A_k(t)$ denotes the unique element
of $\mathcal{A}_k$ containing $t$.
\begin{definition}\cite[p.~25]{2014:talagrand}
  For any $\alpha>0$, define the \emph{Talagrand functional} as:
  \begin{equation}
    \gamma_{\alpha}(\mathcal{T}',d_n)\coloneqq\inf_{\mathcal{A}_k}\sup_{t\in \mathcal{T}'}\sum_{k=0}^\infty 2^{k/\alpha}\Delta_n(A_k(t))
  \end{equation}
  where the infimum is taken over all admissible sequences.
\end{definition}
The Talagrand functional is not necessarily finite.

By $\sigma_n$, we denote the minimal standard deviation of $Z_n(t)$, i.e.
\begin{equation}
  \sigma_n^2\coloneqq\inf_{t\in\mathcal{T}'}\var(Z_n(t))
\end{equation}
We can now prove our second main result:
\begin{theorem}\label{thm:bound:bias}
  Assume $\estlow_n(X)\coloneqq\min_{t\in\mathcal{T}}\frac{1}{n}\sum_{i=1}^n \est(X_i,t)$.
  There is a constant $L>0$ such that,
  if $\est_n(X,t)$ is sub-Gaussian,
  then, for all $u>0$,
  \begin{equation}\label{eq:bound:bias}
    P\left(
      |\estlow_n(X)-\lowtheta|> u(\sigma_1+\gamma_2(\mathcal{T}',d_1))
    \right)
    \le L\exp(-\tfrac{nu^2}{2})
  \end{equation}
  and
  \begin{equation}\label{eq:bound:bias2}
    E\left(|\estlow_n(X)-\lowtheta|\right)
    \le L\frac{\sigma_1+\gamma_2(\mathcal{T}',d_1)}{\sqrt{n}}.
  \end{equation}
\end{theorem}

Note that the constant $L$ in the above theorem is universal,
and bounds for it can be computed directly from Talagrand's proof
\cite[see bounds preceding Eq.~(2.31)]{2014:talagrand}.

\begin{proof}
  Because of our continuity assumptions,
  there is an $s\in\mathcal{T}$ such that
  \begin{equation}
    \sigma_n^2=\var(Z_n(s))
  \end{equation}
  Without loss of generality, we can assume that $s\in\mathcal{T}'$.
  If not, just add $s$ to $\mathcal{T}'$.

  Because $\est_n(X,t)$ is sub-Gaussian,
  we have the following bound for some constant $L'>0$
  \cite[see Eq.~(2.31) and use $S\le\gamma_2(\mathcal{T}',d_n)$]{2014:talagrand}:
  \begin{equation}\label{eq:talagrand:exponential:bound}
    P\left(\sup_{t\in\mathcal{T}'}|Z_n(t)-Z_n(s)| > u\gamma_2(\mathcal{T}',d_n)\right)
    \le L'\exp(-\tfrac{u^2}{2})
  \end{equation}
  Consequently,
  \begin{align}
    \label{eq:talagrand:proof:1}
    &P(|\estlow_n(X)-\lowtheta| > u(\sigma_n+\gamma_2(\mathcal{T}',d_n)))
    \\
    &\le
    P\left(\sup_{t\in\mathcal{T}'}|Z_n(t)| > u(\sigma_n+\gamma_2(\mathcal{T}',d_n))\right)
    \\
    &\le
    P\left(|Z_n(s)|+\sup_{t\in\mathcal{T}'}|Z_n(t)-Z_n(s)| > u(\sigma_n+\gamma_2(\mathcal{T}',d_n))\right)
    \\
\intertext{and now using $\{A+B>0\}\subseteq\{A>0\}\cup\{B>0\}$ for appropriate choice of $A$ and $B$,}
    &\le
    P(|Z_n(s)| > u\sigma_n)+P\left(\sup_{t\in\mathcal{T}'}|Z_n(t)-Z_n(s)| > u\gamma_2(\mathcal{T}',d_n)\right)
    \\
    &
    \label{eq:talagrand:proof:2}
    \le
    2\exp(-\tfrac{u^2}{2})
    +
    L'\exp(-\tfrac{u^2}{2})
    \intertext{where we used the Gaussian tail bound,
      $P(|Z|>z)\le 2 e^{-z^2/2}$ for standard Gaussian $Z$,
      and also \cref{eq:talagrand:exponential:bound}; now
      choose $L\coloneqq\sqrt{\frac{\pi}{2}}(L'+2)$ to arrive at}
    &\le
    L\exp(-\tfrac{u^2}{2})
  \end{align}

  Finally, use $\sigma_n=\frac{1}{\sqrt{n}}\sigma_1$
  and $\gamma_2(\mathcal{T}',d_n)=\frac{1}{\sqrt{n}}\gamma_2(\mathcal{T}',d_1)$
  to arrive at \cref{eq:bound:bias}.
  The $\sigma_n$ equality follows from the usual properties of the
  variance of a sum of i.i.d.\ variables.
  The Talagrand functional equality follows if we can show that
  $d_n(s,t)=\frac{1}{\sqrt{n}}d_1(s,t)$.

  Indeed, first, let $W_i(t)\coloneqq \est(X_i,t)-\theta(t)$, and note that
  \begin{align}
    n^2 (Z_n(s)-Z_n(t))^2
    &=\left(\sum_{i=1}^n W_i(s)-W_i(t)\right)^2
    \\
    &=\sum_{i=1}^n\sum_{j=1}^n (W_i(s)-W_i(t))(W_j(s)-W_j(t))
    \\
    &=\sum_{i=1}^n\left(W_i(s)-W_i(t)\right)^2
    \\
    &\quad
      +2\sum_{i=2}^n\sum_{j=1}^{i-1} (W_i(s)-W_i(t))(W_j(s)-W_j(t))
  \end{align}
  After taking expectations on both sides,
  \begin{align}
    n^2 E\left((Z_n(s)-Z_n(t))^2\right)
    &=\sum_{i=1}^n E\left(\left(W_i(s)-W_i(t)\right)^2\right)
    \\
    &\quad
      +2\sum_{i=2}^n\sum_{j=1}^{i-1} E(W_i(s)-W_i(t))E(W_j(s)-W_j(t))
    \\
    &=n E\left(\left(W_1(s)-W_1(t)\right)^2\right)
    \\
    &=n E\left((Z_1(s)-Z_1(t))^2\right)
\end{align}
where we used the fact that $W_i(s)-W_i(t)$ and $W_j(s)-W_j(t)$ are
independent for $i\neq j$, and that $(W_i(s)-W_i(t))_{i=1}^n$ are
i.i.d.\ and have expectation zero.
Using the definition of $d_n(s,t)$, we conclude that
\begin{equation}
  d_n(s,t)=\frac{1}{\sqrt{n}}d_1(s,t)
\end{equation}
as desired.

To see why the inequality in \cref{eq:bound:bias2} holds, observe that
for any non-negative random variable $V$
\begin{equation}
  P(V>\alpha u)\le \beta\exp(-\tfrac{u^2}{2})
\end{equation}
implies
\begin{align}
  E(V)
  =\int_0^{\infty}P(V>\alpha u)\alpha\,\mathrm{d}u
  \le\alpha\beta\int_0^{\infty}\exp(-\tfrac{u^2}{2})\,\mathrm{d}u
  =\alpha\beta\sqrt{\frac{\pi}{2}}.
\end{align}
Now choose $V$ and $\alpha$ as in \cref{eq:talagrand:proof:1},
choose $\beta=2+L'$,
and apply the inequality established in \cref{eq:talagrand:proof:2}.
\end{proof}

We immediately conclude:

\begin{theorem}
  Assume $\estlow_n(X)\coloneqq\min_{t\in\mathcal{T}}\frac{1}{n}\sum_{i=1}^n \est(X_i,t)$.
  If $\est_n(X,t)$ is sub-Gaussian,
  then $\estlow_n(X)$
  is a consistent estimator for $\lowtheta$
  whenever the minimal standard deviation $\sigma_1$
  and the Talagrand functional $\gamma_2(\mathcal{T}',d_1)$
  are finite.
\end{theorem}

To establish whether or not $\gamma_2(\mathcal{T}',d_1)$ is finite,
a range of practical lower and upper bounds
are provided in \cite[Section~2.3 et. seq.]{2014:talagrand}.
On a very basic intuitive level, we want
\begin{align}
  d_1(s,t)^2
  &=E\left(
    \left(
      \est(X,t)-\theta(t)-\est(X,s)+\theta(s)
    \right)^2
  \right)
  \\
  &=
  \var(\est(X,s))+\var(\est(X,t))-2\cov(\est(X,s),\est(X,t))
\end{align}
to be `small'. This happens precisely when the estimators
$\est(X,s)$ and $\est(X,t)$ are highly correlated for all $s$ and $t$.
A simple but important practical example
where $d_1$ is `too large' is given next:
\begin{theorem}
  There is a constant $M>0$ such that,
  if for some $\epsilon>0$ we have that $d_1(s,t)\ge\epsilon$ for all
  $s\neq t$, then 
  \begin{equation}
    \gamma_2(\mathcal{T}',d_1)\ge M\epsilon\sqrt{\log m}
  \end{equation}
  where $m$ is the cardinality of $\mathcal{T}'$. 
\end{theorem}
\begin{proof}
  Immediate from \cite[Theorem~2.4.1]{2014:talagrand} (the majorizing
  measure theorem) and \cite[Lemma~2.4.1]{2014:talagrand} (Sudakov
  minoration).
\end{proof}
In particular, if $\mathcal{T}'$ is not finite, then $\gamma_2(\mathcal{T}',d_1)=\infty$
under the conditions of the theorem.
The condition $d_1(s,t)\ge\epsilon$ for all $s\neq t$
obtains for instance when all
$\est(X_1,t)$ are pairwise independent
and $\min_{t\in\mathcal{T}}\var(\est(X_1,t))>0$.
Such an estimator will perform very badly.
For example, this tells us that
when doing two-level Monte Carlo,
one should fix the random seed for every run,
in order to ensure that the different runs are maximally correlated
(and definitely not independent!).

\subsection{Confidence Interval}
\label{sec:confidence}

In cases where we can bound the Talagrand functional,
\cref{eq:bound:bias} in \cref{thm:bound:bias} can be used directly
to construct a confidence interval for $\lowtheta$ around $\estlow_n(X)$.
In general, however, bounding the Talagrand functional is a non-obvious
procedure.
The next theorem provides a much simpler procedure
for constructing a confidence interval for $\lowtheta$,
using i.i.d.\ realisations
of $\estlow(X)$ and $\estupp(X,X')$ instead of using the Talagrand functional.
The price we pay is that we need to repeat our calculation
for a sufficient number of realisations of $X$.

Note a crucial difference in notation from \cref{sec:consistency}:
there $X$ was a vector of random variables $(X_1,X_2,\dots)$.
In the theorem below, we need to work with independent realisations of $X$,
where $X$ might be a vector as before, or something else.
Either way,
to avoid possible confusion with the components of $X$, we will denote
these independent realisations by $\chi_1$, $\chi_2$, \dots, $\chi_N$,
and so on.

To apply the central limit theorem, we assume below that $\est(X,t)$ is
uniformly bounded, but obviously this can be relaxed in the usual ways
\cite{1987:nelson}.

\begin{theorem}\label{thm:confint}
  Let $\chi_1$, \dots, $\chi_N$, $\chi'_1$, \dots, $\chi'_N$ be a sequence of
  i.i.d.\ realisations of $X$. Define
  \begin{align}
    Y_*&\coloneqq (\estlow(\chi_i))_{i=1}^N \\
    Y^*&\coloneqq (\estupp(\chi_i,\chi_i'))_{i=1}^N
  \end{align}
  Let $\bar{Y}_*$ and $\bar{Y}^*$ be the sample means of these sequences,
  and let $S_*$ and $S^*$ be their sample standard deviations.
  Let $t_{N-1}$ denote the usual two-sided critical value of the t-distribution
  with $N-1$ degrees of freedom at confidence level $1-\alpha$.
  Then, provided
  that $\sup_{x,t}|\est(x,t)|<+\infty$,
  \begin{equation}
    P\left(
    \bar{Y}_*-t_{N-1}\frac{S_*}{\sqrt{N}}
    \le\lowtheta\le
    \bar{Y}^*+t_{N-1}\frac{S^*}{\sqrt{N}}
    \right)
    \ge
    1-\alpha.
  \end{equation}
  In other words,
  $[\bar{Y}_*-t_{N-1}\frac{S_*}{\sqrt{N}},
    \bar{Y}^*+t_{N-1}\frac{S^*}{\sqrt{N}}]$
  is an approximate confidence interval for $\lowtheta$
  with confidence level (at least) $1-\alpha$.
\end{theorem}

Before we proceed with the proof,
it is worth to make a note about computational efficiency.
The slowest part of the computation normally is the optimisation
over $t$, i.e. the evaluation of $\tau$.
To find the confidence interval, we need to run $2N$
evaluations of $\tau$: one for each $\tau(\chi_i)$ in $Y_*$
and one for each $\tau(\chi'_i)$ in $Y^*$.
However, instead of using
$\estupp(\chi_i,\chi'_i)=\est(\chi_i,\tau(\chi'_i))$,
it will be much faster to use
$\estupp(\chi'_i,\chi_i)=\est(\chi'_i,\tau(\chi_i))$, because
for the latter we can recycle the already computed value of $\tau(\chi_i)$.
This still produces a valid confidence interval
for $\estupp(X,X')$.
However,
because it may happen that $\estupp(x',x)<\estlow(x)$ for some realisations
of $x$ and $x'$,
it is not guaranteed
that $\bar{Y}_*\le\bar{Y}^*$ with this modification.
Consequently, the resulting confidence interval might be empty.
However, the proof below only relies on the probability of
\begin{equation}
    \left\{
    \bar{Y}_*-t_{N-1}\frac{S_*}{\sqrt{N}}
    >\lowtheta
    \right\}
    \cap
    \left\{
    \lowtheta>
    \bar{Y}^*+t_{N-1}\frac{S^*}{\sqrt{N}}
    \right\}
\end{equation}
being close to zero.
Since this is very likely to be the case in practice, we suggest to use this faster
method, even if it has a minor theoretical flaw,
because it will be about twice as fast.
It should only not be used when you find that your confidence
intervals are regularly empty.

\begin{proof}
  By \cref{thm:lowuppest}, and the central limit theorem,
  \begin{align}
    P\left(
    \bar{Y}_*-t_{N-1}\frac{S_*}{\sqrt{N}}
    \le E(\estlow(X))
    \right)
    \simeq
    1-\alpha/2,
    \\
    P\left(
    E(\estupp(X,X'))\le
    \bar{Y}^*+t_{N-1}\frac{S^*}{\sqrt{N}}
    \right)
    \simeq
    1-\alpha/2,
  \end{align}
  where $X'$ is as before an i.i.d.\ realization of $X$.
  So, because $E(\estlow(X))\le\lowtheta\le E(\estupp(X,X'))$,
  \begin{align}
    P\left(
    \bar{Y}_*-t_{N-1}\frac{S_*}{\sqrt{N}}
    \le \lowtheta
    \right)
    \ge
    1-\alpha/2,
    \\
    P\left(
    \lowtheta\le
    \bar{Y}^*+t_{N-1}\frac{S^*}{\sqrt{N}}
    \right)
    \ge
    1-\alpha/2.
  \end{align}
  We also know that
  \begin{multline}
    \left\{
    \bar{Y}_*-t_{N-1}\frac{S_*}{\sqrt{N}}
    >\lowtheta
    \right\}
    \cup
    \left\{
    \lowtheta>
    \bar{Y}^*+t_{N-1}\frac{S^*}{\sqrt{N}}
    \right\}
    \\
    =
    \left\{
    \bar{Y}_*-t_{N-1}\frac{S_*}{\sqrt{N}}
    \le\lowtheta\le
    \bar{Y}^*+t_{N-1}\frac{S^*}{\sqrt{N}}
    \right\}^c
  \end{multline}
  and moreover, the intersection of the sets
  on the left hand side of the above expression must be empty,
  because it is guaranteed that $\bar{Y}_*\le\bar{Y}^*$
  by \cref{eq:estlowupp:ineq}.
  Consequently,
  \begin{multline}
    \alpha/2+\alpha/2
    \ge
    P\left(
    \bar{Y}_*-t_{N-1}\frac{S_*}{\sqrt{N}}
    > \lowtheta
    \right)
    +
    P\left(
    \lowtheta >
    \bar{Y}^*+t_{N-1}\frac{S^*}{\sqrt{N}}
    \right)
    \\
    =
    1-P\left(
    \bar{Y}_*-t_{N-1}\frac{S_*}{\sqrt{N}}
    \le\lowtheta\le
    \bar{Y}^*+t_{N-1}\frac{S^*}{\sqrt{N}}
    \right)
  \end{multline}
\end{proof}

\subsection{Confidence Interval for Biased Estimators}

We briefly consider the case where $\est(X,t)$ is biased.
This will allow us to apply our results also on estimators that
are only asymptotically unbiased. Specifically, let $\est(X,t)$ be any
estimator for $\theta(t)$ such that
\begin{equation}
  \left|E(\est(X,t))-\theta(t)\right|\le \beta,
\end{equation}
for some constant $\beta$ which does not depend on $t$.
We then have the following result,
extending \cref{thm:confint}.

\begin{theorem}\label{thm:confint2}
  Let $\chi_1$, \dots, $\chi_N$, $\chi'_1$, \dots, $\chi'_N$ be a sequence of
  i.i.d.\ realisations of $X$. Define
  \begin{align}
    Y_*&\coloneqq (\estlow(\chi_i))_{i=1}^N \\
    Y^*&\coloneqq (\estupp(\chi_i,\chi_i'))_{i=1}^N
  \end{align}
  Let $\bar{Y}_*$ and $\bar{Y}^*$ be the sample means of these sequences,
  and let $S_*$ and $S^*$ be their sample standard deviations.
  Let $t_{N-1}$ denote the usual two-sided critical value of the t-distribution
  with $N-1$ degrees of freedom at confidence level $1-\alpha$.
  Then, provided
  that $\sup_{x,t}|\est(x,t)|<+\infty$,
  \begin{equation}
    P\left(
    \bar{Y}_*-t_{N-1}\frac{S_*}{\sqrt{N}}-\beta
    \le\lowtheta\le
    \bar{Y}^*+t_{N-1}\frac{S^*}{\sqrt{N}}+\beta
    \right)
    \ge
    1-\alpha.
  \end{equation}
\end{theorem}
\begin{proof}
  We know that, for all $t\in\mathcal{T}$,
  \begin{equation}
    \theta(t)-\beta\le E(\est(X,t))\le \theta(t)+\beta
  \end{equation}
  and therefore
  \begin{equation}
    \lowtheta-\beta\le\min_{t\in\mathcal{T}}E(\est(X,t))\le\lowtheta+\beta
  \end{equation}
  since $\min_{t\in\mathcal{T}}\theta(t)=\lowtheta$.
  Consequently,
  \begin{multline}
    \left\{
    \bar{Y}_*-t_{N-1}\frac{S_*}{\sqrt{N}}-\beta
    \le\lowtheta\le
    \bar{Y}^*+t_{N-1}\frac{S^*}{\sqrt{N}}+\beta
    \right\}
    \\
    \supseteq
    \left\{
    \bar{Y}_*-t_{N-1}\frac{S_*}{\sqrt{N}}
    \le\min_{t\in\mathcal{T}}E(\est(X,t))\le
    \bar{Y}^*+t_{N-1}\frac{S^*}{\sqrt{N}}
    \right\}
  \end{multline}
  Now note that $\est(X,t)$ is an unbiased estimator for $E(\est(X,t))$,
  and apply \cref{thm:confint}.
\end{proof}

\subsection{Coherence of the Lower Estimator}

To end this section on estimators for the minimum of a function,
we study the case where we try to estimate a lower prevision,
and in particular, we analyze under which conditions
the estimate $\estlow(X)$
produces a coherent lower prevision.
This is highly desirable, for instance
if the estimate is consequently used for decision making,
because algorithms and properties of decision rules
typically rely heavily on coherence \cite{2007:troffaes:decision:intro}.
In particular, we would not want $\estlow(X)$ to incur a sure loss.

Consider an arbitrary set $\gambles$ of random quantities.
Let $P_t$ be a linear prevision (i.e. expectation operator)
on $\gambles$, parameterized by $t\in\mathcal{T}$.
For every $f$ in $\gambles$,
the lower prevision (or, lower expectation) of $f$ is defined as
\begin{equation}
  \underline{P}(f)\coloneqq\min_{t\in\mathcal{T}}P_t(f).
\end{equation}
We now assume that we have an estimator for each $f$ in $\gambles$.
Equivalently, we have an estimator defined on
$\est\colon\mathcal{X}\times\mathcal{T}\times\gambles$,
such that for every $f\in\gambles$ and $t\in\mathcal{T}$:
\begin{equation}
  E(\est(X,t,f))=P_t(f)
\end{equation}
so each $\est(X,t,f)$ is an estimator for $P_t(f)$.
Under suitable conditions, we know that
\begin{equation}
  \estlow(X,f)\coloneqq\min_{t\in\mathcal{T}}\est(X,t,f)
\end{equation}
is a consistent estimator for $\underline{P}(f)$.
What we would like additionally, however, is
for any possible realization $x$ of $X$,
$\estlow(x,\cdot)$ (as a function on $\gambles$)
to be a coherent lower prevision as well.
The next theorem settles this question.

\begin{theorem}\label{thm:coherence}
  Let $x$ be any realization of $X$.
  If, for every $t\in\mathcal{T}$,
  $\est(x,t,\cdot)$ is a coherent prevision,
  then $\estlow(x,\cdot)$ is a coherent lower prevision.
\end{theorem}

The proof follows from the fact that the lower envelope of coherent
previsions always produces a coherent lower prevision.

The importance of the theorem is that, in order for our estimate
$\estlow(x,\cdot)$ to form a coherent lower prevision,
we need each individual estimator
$\est(x,t,\cdot)$ to be a coherent prevision.

The other importance is a conceptual one:
to maintain coherence,
we use the same realization $x$ of $X$
for all $f$ in $\gambles$.
Coincidently, this will also reduce computational effort,
because we only need to sample once to obtain an estimate
for the entire lower prevision.

\section{Iterated Importance Sampling}
\label{sec:iteratedimpsamp}

In this section, we apply the theory developed in \cref{sec:estimation}
on a specific estimator,
namely the one that is associated with importance sampling.
We follow the treatment presented earlier in \cite{2017:troffaes::impmc}.
The main difference is that in this paper,
we are in a position to provide proper confidence intervals.
We can also ask deeper research questions about the theoretical properties
of importance sampling for estimating lower previsions.

\subsection{Importance Sampling Estimates}
\label{sec:impsam}

We first review the basic ideas behind importance sampling.
For the theory behind the results that are presented here,
we refer to \cite[Chapter~9]{2013:owen::mcbook}.

Consider a parameterized collection
of probability density functions
$p_t$ on $\mathcal{X}$.
Assume we have an i.i.d.\ sample $x_1$, \dots, $x_n$
drawn from a strictly positive probability density function $q$
on $\mathcal{X}$.
For example, we could have that $q=p_{\tilde{t}}$ for some fixed $\tilde{t}\in\mathcal{T}$, but we do not require this.
Throughout the entire paper, we will consider many different probability
density functions, but the sample $x_1$, \dots, $x_n$ will always be one
drawn from $q$.
Assume we have a real-valued function $f\colon\mathcal{X}\to\reals$,
and we would like to
estimate the expectation of $f$ with respect to $p_t$,
for some arbitrary choice of $t\in\mathcal{T}$.

In case $p_t=q$, by the central limit theorem,
an approximate 95\% confidence interval
for the expectation of $f$ with respect to $q$
is then given by $\est(x)\pm 1.96 \estsigma(x)/\sqrt{n}$ where
\begin{align}
\est(x)&\coloneqq\frac{1}{n}\sum_{i=1}^n f(x_i) &
\estvar(x)&\coloneqq\frac{1}{n-1}\sum_{i=1}^n (f(x_i)-\est(x))^2
\end{align}

Can we use the same sample $x_1$, \dots, $x_n$ drawn from $q$ to
get an estimate for the expectation of $f$ with respect to $p_t$
for \emph{any} $t$, when $p_t\neq q$?
The following equality gives a clue
as to how we might do that:
\begin{equation}
  \int f(x)p_t(x)dx=\int \frac{p_t(x)}{q(x)}f(x)q(x)dx=\int w_t(x)f(x)q(x)dx
\end{equation}
where $w_t=p_t/q$.
So, the expectation of $f$ with respect to $p_t$ is the same as
the expectation of
$w_t f$ with respect to $q$, and therefore
an approximate 95\% confidence interval
for the expectation of $f$ with respect to $p_t$
is then given by $\est(x,t)\pm 1.96 \estsigma(x,t)/\sqrt{n}$ where
\begin{align}
\est(x,t)&\coloneqq\frac{1}{n}\sum_{i=1}^n w_t(x_i) f(x_i) \\
\estvar(x,t)&\coloneqq\frac{1}{n-1}\sum_{i=1}^n (w_t(x_i) f(x_i)-\est(x,t))^2
\end{align}
This estimate is called the \emph{importance sampling estimate}.

The estimate $\est(x,t)$, when seen as a function of $f$,
does not produce a coherent prevision, because the weights $w_t(x_i)$
will usually not sum to $n$.
Additionally,
the normalisation constant of the densities is often unknown, or is
slow to compute, and we only know $w'_t=c p_t/q$ for
some unknown value of $c$. We can address both of these issues
by using the
\emph{self-normalised importance sampling estimate}:
\begin{align}
  \label{eq:impsamestimate}
  \est(x,t)&\coloneqq\frac{\sum_{i=1}^n w'_t(x_i) f(x_i)}{\sum_{i=1}^n w'_t(x_i)}
  \\
  \estvar(x,t)&\coloneqq\frac{1}{n-1}
    \frac{%
       \frac{1}{n}\sum_{i=1}^n w'_t(x_i)^2 (f(x_i)-\est(x,t))^2}{%
       \left(\frac{1}{n}\sum_{i=1}^n w'_t(x_i)\right)^2}
\end{align}
This estimator, when seen as a function of $f$, does produce
a coherent prevision, and therefore the associated
lower estimator $\estlow(x)$ will produce a coherent lower prevision
(see \cref{thm:coherence}).

There are some downsides to using the self-normalised importance
sampling estimate.  First of all, unlike the standard importance
sampling estimator, the self-normalised estimator has a bias of order
$O(1/n)$, i.e. it is only asymptotically unbiased.  Although this bias
converges to zero as the sample size $n$ increases, for small sample
sizes, the bias can be substantial, and there is no guarantee that
the bias is uniformly bounded as we vary $t$.  We do note that the
bias becomes zero if the weights are constant, i.e. if the sampling
distribution $q$ is close to the target distribution $p_t$. This will
be important later when we consider the iterative procedure.

Secondly,
because \cref{eq:impsamestimate} does not have the form of a sample mean,
the Glivenko-Cantelli class condition
and results for the sub-Gaussian case cannot be applied here.
Those results only apply to the standard importance sampling estimate.
In particular, if the set of functions
\begin{equation}
  \{w_t(X)f(X)\colon t\in\mathcal{T}'\}
\end{equation}
form a Glivenko-Cantelli class, then the lower envelope of the
standard importance sampling estimators will be consistent.

Although $\estvar(x,t)$ gives an indication of the quality of the
estimate, one must be wary that $\estvar(x,t)$ is by itself only an
approximation of the true error. An additional diagnostic to consider
is the effective sample size, which can be calculated as follows:
\begin{equation}
  n(x,t)\coloneqq\frac{\left(\sum_{i=1}^n w'_t(x_i)\right)^2}{\sum_{i=1}^n w'_t(x_i)^2}
\end{equation}
Note that there are many different ways to define effective sample
size and even more ways to define diagnostics for importance
sampling. What matters for this paper is that a low $n(x,t)$ is bad, and
that $n(x,t)\simeq n$ is good. For an in-depth discussion about
diagnostics for importance sampling, we refer to
\cite[Section~9.3]{2013:owen::mcbook}.

\subsection{Imprecise Importance Sampling Estimates}
\label{sec:impreciseimpsam}

Importance sampling has many different uses, including
variance reduction, numerical integration, and Bayesian inference.
Here, we will use the theory developed in \cref{sec:estimation}
in order to
estimate the lower prevision of a gamble $f$.
O'Neill \cite{2009:oneill:importancesampling} studied this technique already
in a Bayesian setting, although only studying the consistency of
his estimator $\est(X,t)$ and not the bias and consistency of $\estlow(X)$
as we do here.

Given our parameterized collection of
probability density functions $p_t$,
we define the \emph{lower prevision} of $f$ as
\begin{equation}\label{eq:lowprev}
  \underline{P}(f)\coloneqq\min_{t\in\mathcal{T}}\int f(x)p_t(x)\,\mathrm{d}x
\end{equation}
where, as in \cref{sec:estimation},
we assume that the minimum is achievable, for simplicity of presentation.
With, as before,
\begin{equation}
\tau(x)\coloneqq\arg\min_{t\in\mathcal{T}}\est(x,t)
\end{equation}
where $\est(x,t)$ is the importance sampling estimate,
we know from \cref{thm:lowuppest} that
\begin{align}
  \estlow(x)&=\est(x,\tau(x)) & \estupp(x,x')&=\est(x,\tau(x'))
\end{align}
will provide lower and upper estimates for $\underline{P}(f)$.
The key observation here is
that we can use the same $x_1$, \dots, $x_n$, $x_1'$, \dots, $x_n'$
across all choices of $t\in\mathcal{T}$, and that the
optimisation procedure operates on the weights only.
Additionally,
the importance sampling
estimates $\est(X,t)$ will be correlated across different values of $t$.
As discussed in \cref{sec:discrepancy:bounds},
the stronger this correlation is,
the lower will be the value of the Talagrand functional,
and therefore the lower will be the bias of $\estlow(X)$.

If we repeat the estimation $N$ times
(so, we need $N\times 2n$ i.i.d.\ samples from $q$ in total),
we can construct a
confidence interval for $\underline{P}(f)$
using \cref{thm:confint},
provided for instance that $f$ is bounded
so that the central limit theorem applies.

Note that in the self-normalised case,
by \cref{thm:confint2}, the confidence interval as constructed
in \cref{thm:confint} is still approximately correct for large $n$,
because this estimator is still asymptotically unbiased,
provided that the bias is uniformly bounded as we vary $t$.

One issue with this method is that the standard error $\estsigma(x,t)$ can be
very large, especially if $p_t$ is very different from $q$.
So, the method will only work efficiently if
$\estsigma(x,t)$ remains
reasonably bounded across $\mathcal{T}$.
From the literature on importance sampling
for variance reduction, we know that good choices for $q$ are
those that are proportional to $|f|p_t$
\cite[Chapter~9, p.~6]{2013:owen::mcbook}.
So, in case $p_t$ covers a wide range of distributions,
it may be hard to identify a single sampling distribution $q$.
\cite[Section~3]{2016:zhang:importancesampling} discuss ways of
choosing optimal sampling distributions for credal sets.

A second problem is that, in general, there is no single sampling
distribution $q$ that can guarantee a good effective sample size $n(x,t)$ for
all $t\in\mathcal{T}$. Consequently, with this approach, even if we
try to choose $q$ optimally, the effective sample size at $\tau(x)$ can
still become extremely low.

\subsection{Example}
\label{sec:exam}

As a first example, we demonstrate imprecise importance sampling
on the imprecise Dirichlet model,
similar to the one studied in \cite{2009:oneill:importancesampling}.
The starting point of the example is identical to the one studied
in \cite{2017:troffaes::impmc}.
However, the confidence intervals in \cite{2017:troffaes::impmc} presumed,
as stated,
an unbiased case, which was not exactly satisfied.
Here, we provide better confidence intervals based on the method that
we developed in \cref{sec:confidence}.

Denote the $k$-dimensional unit simplex by $\Delta$.
Consider an unknown parameter $x\in\Delta$,
say, modelling the probabilities of some multinomial process.
Consider the following class of probability density functions
on $x$:
\begin{equation}\label{eq:idmdensity}
p_t(x)=\frac{\Gamma(s)}{\prod_{j=1}^k\Gamma(st_j)}\prod_{j=1}^kx_j^{st_j-1}
\end{equation}
with hyperparameters $s>0$ and $t\in\Delta$---these are Dirichlet distributions.
We are interested in finding the lower expectation of some function
$f(x)$, over all $t\in\mathcal{T}\subseteq\Delta$ and with $s=2$ fixed.

For $q(x)$, we take the Dirchlet distribution with uniform $\tilde{t}_j=1/k$ and
with the same value for $\tilde{s}=2$.

In order to apply importance sampling, we need to calculate the weight function.
The weights are:
\begin{equation}
w_t(x)=p_t(x)/q(x)\propto \prod_{j=1}^kx_j^{st_j-\tilde{s}\tilde{t}_j}=w'_t(x)
\end{equation}
In this case, we have a very simple closed analytical expression for
$w'_t(x)$.  Note that we could also use $w_t(x)$ directly, however
evaluating the normalisation constants requires several evaluations of
the Gamma function, and slows down the optimisation procedure
considerably \cite{2017:troffaes::impmc}.
The optimisation problem can be written as
\begin{equation}\label{eq:optimalt}
  \tau(x)=\arg\min_{t\in\mathcal{T}}
  \frac{\sum_{i=1}^n w'_t(x_i) f(x_i)}{\sum_{i=1}^n w'_t(x_i)}
\end{equation}

As a numerical example, we take $k=5$,
$\mathcal{T}=\{t\in\Delta\colon t_j\ge 0.1\}$, and
$f(x)=x_1+2x_2+5x_3+4x_4-3x_5$. In this case, we know that the exact
expectation of $f$, for fixed $t$, is given by
\begin{equation}
  P_t(f)=t_1+2t_2+5t_3+4t_4-3t_5.
\end{equation}
So, the lower prevision of $f$ over all $t\in\mathcal{T}$ is
clearly achieved for
\begin{equation}\label{eq:example1:optimalt}
  t^*\coloneqq(0.1,0.1,0.1,0.1,0.6)
\end{equation}
and is given by
\begin{equation}
  \underline{E}(f)
  =
  0.1+2\times 0.1+5\times 0.1+4\times 0.1 - 3 \times 0.6
  =
  -0.6
\end{equation}

The simulation code was implemented in R. The \texttt{constrOptim} function
was used to do the actual optimisation, through the downhill simplex method.
\Cref{tab:firstexperiment} summarizes our simulation results.
The 95\% confidence bounds are given by the first two rows of the table.
Besides the confidence interval, as a further diagnostic, we also provide
the mean values for $\tau(x)$ and $n(x,\tau(x))$ across all simulation runs,
or more precisely:
\begin{align}
  \bar{\tau}_*&\coloneqq\frac{1}{N}\sum_{i=1}^N \tau(\chi_i) &
  \bar{n}_*&\coloneqq\frac{1}{N}\sum_{i=1}^N n(\chi_i,\tau(\chi_i))
\end{align}
where each $\chi_i$ represents a different realisation
of the vector $(x_1,\dots,x_n)$, as explained in \cref{sec:confidence}.
The results are presented for the modified (faster) method of obtaining
the upper confidence bound; see the discussion following \cref{thm:confint}.
These bounds were checked against the slower exact theoretical bounds:
the numerical results were nearly identical, but the simulation took
twice as long.

\begin{table}
\centering
\begin{tabular}{r|r|r|r|r|r|r}
$N$       & 4 & 8  & 16 & 32 & 64 & 128 \\
$n$       & 4 & 8  & 16 & 32 & 64 & 128 \\
\hline
$\bar{Y}_*-t_{N-1} S_* / \sqrt{N}$ & -0.522 & -0.884 & -0.445 & -0.603 & -0.614 & -0.560 \\
$\bar{Y}^*+t_{N-1} S^* / \sqrt{N}$ & 4.001 & 1.572 & 1.127 & 0.397 & 0.053 & -0.202 \\
\hline
$\bar{\tau}_{*1}$ & 0.382 & 0.100 & 0.165 & 0.149 & 0.109 & 0.110 \\
$\bar{\tau}_{*2}$ & 0.100 & 0.152 & 0.117 & 0.104 & 0.106 & 0.103 \\
$\bar{\tau}_{*3}$ & 0.176 & 0.108 & 0.110 & 0.100 & 0.102 & 0.101 \\
$\bar{\tau}_{*4}$ & 0.146 & 0.133 & 0.107 & 0.101 & 0.102 & 0.100 \\
$\bar{\tau}_{*5}$ & 0.196 & 0.506 & 0.501 & 0.545 & 0.581 & 0.586 \\
\hline
$\bar{n}_*$ & 1.512 & 2.763 & 3.826 & 6.009 & 11.720 & 15.988 \\
\hline
time (seconds) & 0.636 & 2.008 & 7.601 & 27.925 & 115.052 & 472.733
\end{tabular}
\caption{Importance sampling simulation results for various sample sizes.}
\label{tab:firstexperiment}
\end{table}

If $\estlow(X)$ is theoretically consistent, then we would expect
the lower confidence bound (first row)
to approach the theoretically correct value of $-0.6$ as
$n$ and $N$ increase.
This appears to be the case.
The upper confidence bound improves gradually, but is still very far off.
Either way, the confidence bounds are pretty bad:
the correct value $-0.6$ falls near the lower bound in every case,
and falls outside the interval half of the time.
This is likely partly due to the bias in the self-normalised estimator
for small $n$.

The mean effective sample sizes
$\bar{n}_*$ are a very long way from the full sample size $n$.
This means that the estimators will have large variance,
as is evidenced by the large width of the confidence interval.

Despite the bad estimates, the values for $\tau(x)$, as evidenced by
the $\bar{\tau}_*$ row, are quite close to the theoretically optimal
value (see \cref{eq:example1:optimalt}), even if the effective sample size
is still rather low. So, importance sampling manages to identify the
correct distribution. The issue is that the weights are so skewed
that the sample has poor quality.

In terms of computational time,
the bottleneck is clearly the optimisation procedure.
We emphasize that we have not tried to write the fastest possible code,
and there might still be good opportunities for optimisation.

\subsection{Iterated Importance Sampling}

The example shows that a single importance sampling distribution $q$ may
not provide a good effective sample size across the entire set
of distributions $p_t$, even if $n$ is quite large.
For instance, in the numerical example,
with $n=128$ we still only had $\bar{n}_*\simeq 16$.

What we conclude from this is that plain imprecise
importance sampling does not work very well, even in simple cases.
Next we discuss some extensions of the proposed procedure in order
to make it work better.

Even though the estimates are quite bad,
our numerical experimentation shows that the correct $t^*$, or nearly correct
$t^*$, can be identified already with much lower $n$:
already for $n=8$, the correct value is not too far off,
and for $n=64$, is correct within 10\% relative error.
So, rather than increasing $n$ in order to guarantee
a tighter confidence interval,
one idea is to iterate the procedure
so that $q$ eventually converges to $p_{t^*}$ where $t^*$ is the actual
optimal choice. If $q$ is equal to $p_{t^*}$,
then all weights are identical, $n=n(x,t^*)$, and the self-normalised estimator is no longer biased. Also, in this case,
it turns out that the optimisation in \cref{eq:optimalt} runs very quickly,
because we are already near the optimal solution.

Here is how we might implement this in practice:
\begin{enumerate}[nosep]
\item Set $t$ to a reasonable initial value.
\item\label{item:proc1:generateq}
  Fix the random seed,
  and generate a sample $x_1$, \dots, $x_n$ from $p_t$.
\item Find optimal $t^*$ through \cref{eq:optimalt}.
\item Check if $n_{t^*}$ is close to $n$, or until a maximum number of iterations is reached.
  If so, construct a confidence interval using $p_{t^*}$
  as the sampling distribution (see \cref{thm:confint}), and stop.
\item Set $t=t^*$, and return to \cref{item:proc1:generateq}.
\end{enumerate}
This is the same algorithm as the one presented in
\cite{2017:troffaes::impmc,2017:fetz::montecarlo}, with the difference that
we fix the random seed between iterations
(this removes random effects hindering convergence of the algorithm),
and we calculate a proper confidence interval at the end.

\subsection{Example Revisited}
\label{sec:exam2}

Let us apply the proposed iterative procedure on our Dirichlet example.
For simplicity, we choose a fixed value of $n=128$.
This is also one of the entries in the above table,
so it provides a good basis for comparison.
\Cref{tab:iterations} summarises the results of the first phase
(i.e. the iterations before we construct the confidence interval).
\begin{table}
\centering
\begin{tabular}{r|r|r|r}
iteration & 1 & 2 & 3 \\ 
\hline
$n(x,\tau(x))$ & 2.186 & 123.974 & 128.000 \\ 
\hline
$\tau_1(x)$ & 0.100 & 0.100 & 0.100 \\ 
$\tau_2(x)$ & 0.116 & 0.100 & 0.100 \\ 
$\tau_3(x)$ & 0.100 & 0.100 & 0.100 \\ 
$\tau_4(x)$ & 0.100 & 0.100 & 0.100 \\ 
$\tau_5(x)$ & 0.584 & 0.600 & 0.600 \\ 
\hline
time (seconds) & 2.203 & 1.724 & 1.163
\end{tabular}
\caption{Importance sampling simulation results for each iteration.}
\label{tab:iterations}
\end{table}

This part of the simulation took only 5 seconds.
We see that the simulation converges in just 3 steps.  In the first
step, we get fairly close to the correct $t^*$, even though the
effective sample size $n_t\simeq 2$ is completely off the chart.
The second step
uses this value for $t$ to draw samples, and as this distribution is much
closer to the actual optimal distribution, the effective sample size
increases substantially. In this step, we also identify the correct value
for $t^*$. The last step uses the (nearly) correct distribution for sampling,
and gets a full effective sample size.

To produce a confidence interval, we have two choices.
If we do not want to make any assumptions, we simply apply
\cref{thm:confint}. This will be fairly slow, but we will get
a proper confidence interval. The results are given
in the top part of \cref{tab:iterative:confint}.
\begin{table}
\centering
\begin{tabular}{r|r}
  $N$ & 128 \\ 
  $n$ & 128 \\
\hline 
$\bar{Y}_*-t_{N-1} S_* / \sqrt{N}$ & -0.640 \\ 
$\bar{Y}^*+t_{N-1} S^* / \sqrt{N}$ & -0.585 \\ 
$\bar{\tau}_{*1}$ & 0.100 \\ 
$\bar{\tau}_{*2}$ & 0.100 \\ 
$\bar{\tau}_{*3}$ & 0.100 \\ 
$\bar{\tau}_{*4}$ & 0.100 \\ 
$\bar{\tau}_{*5}$ & 0.600 \\ 
$\bar{n}_*$ & 127.946 \\ 
time (seconds) & 156.113 \\
\hline
$\bar{x}-t_{Nn-1} s / \sqrt{Nn}$ & -0.638 \\ 
$\bar{x}+t_{Nn-1} s / \sqrt{Nn}$ & -0.584 \\ 
time (seconds) & 0.076
\end{tabular}
\caption{Confidence intervals for iterative importance sampling. The top part of the table shows the theoretically correct method, and the bottom of the table shows the fast approximate method under the assumption that the optimisation produces stable values for $t$.}
\label{tab:iterative:confint}
\end{table}
Because the sampling distribution is very close
to the theoretically optimal distribution,
the effective sampling size is almost identical to the actual sampling size
in all of the runs. Moreover, the optimisation itself runs much faster.
The total time taken has reduced from 473 seconds to 161 seconds.

However, we are really wasting a lot of time: if $\tau(x)$ remains
pretty much constant, we can produce a direct confidence interval much faster
simply by sampling directly from $p_t$. For comparability with the
theoretically correct confidence interval,
we use the same total sample size, $Nn=16384$.
The results are presented in the bottom part of \cref{tab:iterative:confint}.
We see that this confidence interval is virtually identical to the
theoretically correct interval. Moreover, it takes less than a tenth
of a second to calculate.  So, if we make a leap of faith and assume
that importance sampling simulations will not deviate away further
from the $t$ that was identified in the last step of the algorithm,
calculating the confidence
interval is really quick, and in this specific case we have reduced
473 seconds down to just 5 seconds.
Sensible criteria for determining whether $t$ will remain stable are:
\begin{itemize}
\item Check that $n(x,t)\simeq n$. If not, then the weights are not evenly
  distributed, and therefore the sampling distribution is unlikely going
  to be the optimal distribution.
\item Rerun the imprecise importance sample (from $t$)
  for a few i.i.d.\ realisations
  $\chi_1$, \dots, $\chi_N$ of $X$,
  and check that $\tau(\chi_i)$ does not deviate much from $t$.
\end{itemize}
Both of these conditions are satisfied in our example. Indeed, by
\cref{tab:iterative:confint}, we see that
$\bar{\tau}_*\simeq t^*$ and $\bar{n}_*\simeq n$. However we may need far
fewer than $N=128$ iterations in order to verify these conditions.

\subsection{Entropy Example}

For comparison, we present our method applied on the importance sampling 
example given by O'Neill \cite[Section~8]{2009:oneill:importancesampling}.
This example also considers the imprecise Dirichlet model
as in \cref{eq:idmdensity}, with $k=2$ (i.e. it is an imprecise Beta model),
$s=10$, and $\mathcal{T}=\{t\in\Delta\colon t_1\ge 0.3,\,t_2\ge 0.6\}$
(the example is presented differently in
\cite{2009:oneill:importancesampling}, but it is equivalent to this one
after some manipulations). We are interested in estimating:
\begin{equation}
  f(x)=-\sum_{i=1}^k x_i\ln(x_i)
\end{equation}
The exact lower expectation for this model is achieved at $t^*=(0.3,0.7)$
and is equal to $\frac{3553}{3600}\simeq 0.5639683$
\cite[Section~3]{2003:hutter:estimation}.
\Cref{tab:entropy:iterations} presents the results from our simulation,
for $n=10^3$ samples (which is the smallest number considered in
\cite[Section~8]{2009:oneill:importancesampling}).
\begin{table}
\centering
\begin{tabular}{r|r|r}
iteration & 1 & 2 \\ 
\hline
$n(x,\tau(x))$ & 873.44201 & 1000.00000 \\ 
\hline
$\tau_1(x)$ & 0.30000 & 0.30000 \\ 
$\tau_2(x)$ & 0.70000 & 0.70000 \\
\hline
  time & 1.61215 & 0.50916
\end{tabular}
\caption{Importance sampling simulation results for each iteration.}
\label{tab:entropy:iterations}
\end{table}

\begin{table}
\centering
\begin{tabular}{r|r}
  $N$ & 10 \\ 
  $n$ & 1000 \\
\hline 
$\bar{Y}_*-t_{N-1} S_* / \sqrt{N}$ & 0.56217 \\ 
$\bar{Y}^*+t_{N-1} S^* / \sqrt{N}$ & 0.56624 \\ 
$\bar{\tau}_{*1}$ & 0.30000 \\ 
$\bar{\tau}_{*2}$ & 0.70000 \\ 
$\bar{n}_*$ & 1000.00000 \\ 
time (seconds) & 5.52760 \\
\hline
$\bar{x}-t_{Nn-1} s / \sqrt{Nn}$ & 0.56133 \\ 
$\bar{x}+t_{Nn-1} s / \sqrt{Nn}$ & 0.56609 \\ 
time (seconds) & 0.03806
\end{tabular}
\caption{Confidence intervals for iterative importance sampling. The top part of the table shows the theoretically correct method, and the bottom of the table shows the fast approximate method under the assumption that the optimisation produces stable values for $t$.}
\label{tab:entropy:confint}
\end{table}

We see that the method converges after just two steps.
As the optimisation problem is only one-dimensional,
the calculation is fast,
even with the large sample size.

Confidence intervals are presented in \cref{tab:entropy:confint}.
With $N=10$ we get a pretty decent confidence interval,
due to the large sample size $n$ for the importance sampling estimates.
Errors are similar to those reported in
\cite[Section~8]{2009:oneill:importancesampling}
for sample size $Nn=10^4$.
Note that, due to the low dimensionality,
we have that $\tau(x)=t^*$ on every iteration, as can be seen from the table.

\subsection{Open Questions}
\label{sec:openquestions}

We end this section with some open questions.

Will the effective
sample size always increase on successive iterations? All numerical
experiments so far studied confirm that this is the case, but it would
be great if we could prove it.
Moreover, if the effective sample size always increases,
will it always converge to $n$, or at least have a high probability to
come within a close distance of $n$?

Under what conditions will
imprecise importance sampling produce stable values for $t$? If we
could identify those conditions, our calculation of confidence intervals
may not require any further optimisation steps.

As we saw, the final error is essentially controlled by the total sample size,
$n\times N$.
How should we optimally choose $n$ and $N$, for a given total sample size?
Large $n$ makes for more accurate estimates of the expectation,
and therefore we expect that it will be easier to identify the optimal
$t^*$ with larger $n$---this is confirmed by our numerical experiments.
So, probably, it is prudent to choose $n$ larger than $N$, and
ideally at least large enough to allow stable estimates (in the sense of $t$).
Currently, we do not know how to do that.

In relation to \cref{sec:estimation}, a theoretical question is under
what circumstances will
the imprecise importance sampling estimate $\estlow(X)$ be consistent?
Standard (not self-normalised) importance sampling
is written in the form of a sample mean, so in that case
$\estlow(X)$ is consistent
if the importance sampling estimators $\est(X,t)$
form a Glivenko-Cantelli class.
In general, the confidence intervals that we obtained in our numerical examples
indicate that all the examples we investigated have
good consistency properties (well, at least for $\estlow(X)$,
not necessarily for $\estupp(X,X')$).
It would be nice if we had a simple method for
deriving bounds on the Talagrand functional
for standard imprecise importance sampling,
for instance using some of the methods
described in \cite[Section~2.3 et. seq.]{2014:talagrand}.
It would be even nicer if we could
establish sufficient conditions for consistency
of self-normalised imprecise importance sampling.

\section{Conclusion}
\label{sec:conclusion}

We developed a theoretical framework
for estimating the minimum of an unknown function, given
an estimator for that function.
This allowed us then to study estimation of lower previsions.
For estimators
that have the form of a sample mean,
we identified the Glivenko-Cantelli condition as
a sufficient condition for consistency,
and we bounded the bias in the sub-Gaussian case by means of Talagrand's functional,
and identified when consistency fails.

The theoretical bounds, based on Talagrand's functional, are not practical
in the sense that there is no quick and easy way to evaluate them
analytically.
Therefore,
we proposed a new `upper' estimator,
which can be used along with the standard
lower envelope estimator.
This allowed us to construct a practical confidence interval.
Unfortunately, we could not say much about the consistency
of this upper estimator,
and our numerical experiments confirmed that this estimator
may not perform too well in practice:
in some of our importance sampling examples, the upper bound estimate
was far more conservative than the lower bound.
However, we identified one situation under which the upper estimator
is unbiased: this is precisely when the lower bound provides
`stable' estimates, in the sense that we described.

As a specific application of this theoretical framework,
we then described how envelopes of importance sampling estimates
can be used to estimate lower previsions. The key observation that
makes this possible is that importance sampling allows us to estimate
means not just from the distribution that we are sampling from, but
from an entire neighbourhood of distributions around the sampling
distribution. Through straightforward optimisation over the importance
sampling weights, we can therefore estimate lower previsions without
having to, say, draw samples from all extreme points of the credal
set. This technique is simple
and is readily applicable for medium sized problems.

We saw that, especially in higher dimensions, taking envelopes over the
weights may not work very well, due to poor effective sample sizes
especially when the optimal distribution is far away from the sampling
distribution. We revisited
the iterative
procedure proposed in \cite{2017:troffaes::impmc,2017:fetz::montecarlo},
which naturally moves the sampling distribution towards the
optimal distribution. We demonstrated how this led to a much quicker
estimate with far less computational power required.
In this paper, we also studied the confidence intervals that arise
from this method, and we saw that we can compute them very quickly,
at least if the procedure is `stable' in a specific sense.

Whilst the procedure that we have described will work well
for medium sized problems, we foresee that for really large scale problems,
the effective sample size may still be too limited to ensure
that the optimal distribution can be identified at all.
In such cases, perhaps the credal set could scale throughout the algorithm,
in order to ensure a reasonable effective sample size, and therefore to help
convergence of the algorithm.

Another idea is to use importance sampling to explore only a very
small region of the credal set, but then to use the resulting
derivative information to move the sampling distribution
in the right direction.  A problem
with this however is that the derivatives obtained are quite noisy,
and in practice we have not found a good way of using these noisy
derivatives to ensure convergence.

Despite the many open questions listed in \cref{sec:openquestions},
iterative importance sampling for lower
previsions seems promising. Even if we cannot answer the above
questions yet, we can obtain accurate confidence intervals
for envelopes of more or less arbitrary parameterized sets of densities
and arbitrary gambles.
We are thus, in principle,
no longer restricted to specific gambles or dependent on conjugacy
or other special analytical properties of our credal set
in order to work with lower previsions.

Finally, we emphasize once more that
the results from \cref{sec:estimation} are not just applicable
to the estimation of lower previsions,
but to arbitrary envelopes of sets of estimators.
This might potentially be useful in other fields, for instance for
the optimisation of functions approximated by an emulator,
provided the emulator is sub-Gaussian.

\section*{Acknowledgements}

This work is partially supported by the H2020 Marie Curie ITN,
UTOPIAE, Grant Agreement No. 722734.
The author also thanks both reviewers for their insightful and
constructive comments.

\bibliographystyle{plain}
\bibliography{all}

\end{document}